\documentclass[conference]{IEEEtran}
\usepackage{cite}
\usepackage{amsmath,amssymb,amsthm,mathtools}
\usepackage{booktabs,array,multirow}
\usepackage{enumitem}
\usepackage{url}
\usepackage{xcolor}
\usepackage{microtype}
\usepackage{balance}
\usepackage{graphicx}
\usepackage{tikz}
\usepackage{placeins}
\usepackage{stfloats}
\usepackage{hyperref}
\hypersetup{hidelinks}

\newtheorem{theorem}{Theorem}
\newtheorem{lemma}[theorem]{Lemma}
\newtheorem{proposition}[theorem]{Proposition}

\theoremstyle{definition}
\newtheorem*{schemedefinition}{Definition (Graph-CMSS scheme)}
\theoremstyle{plain}

\newcommand{\F}{\mathbb F}
\newcommand{\Rel}{\operatorname{Rel}}

\newcommand{\wt}{\operatorname{wt}}
\newcommand{\Unif}{\operatorname{Unif}}
\newcommand{\bits}{\{0,1\}}
\newcommand{\xor}{\oplus}

\newcommand{\cO}{\mathcal O}
\newcommand{\cG}{\Gamma_4}
\newcommand{\cU}{\mathsf U}
\newcommand{\cQ}{\mathsf Q}
\newcommand{\Z}{\mathbb Z}

\newcommand{\turn}{\operatorname{turn}}
\newcommand{\View}{\mathsf{View}}

\title{The Price of Privacy: Randomness Complexity of Graph-Based Multi-Secret Sharing}

\author{Piotr Marszalik\\[2mm]\small Institute of Science and Technology Austria\\[2mm]\small University of Warsaw\\[2mm]\small \texttt{pm448533@students.mimuw.edu.pl}\\[2mm]\small October 2026}

\begin{document}
% Single-file K4 classification with introductory intuition and a main-text proof roadmap.
% Appendices retain only proof material and computations for the K4 classification
% and its exact one-bit prerequisite. Equation numbering is automatic.
\maketitle

\begin{abstract}
We study the randomness required to share possibly correlated secret bits
among parties connected by a graph. A dealer places shares on the edges
so that each party can recover its own secret from its incident shares
and learn nothing about the others beyond what its own secret reveals.
Anilkumar et al.\ completely determined the minimum randomness required
for three binary secrets. We extend this study to four secrets and obtain
results for arbitrary numbers of secrets on general graphs. For four
parties on a complete graph, we determine the minimum number of random
states for every set of permitted secret combinations: the possible
values are one, two, three and four. We also characterize when one random
bit suffices on an arbitrary graph.
\end{abstract}

\section{Introduction}

A dealer places shares on the edges of a graph whose vertices carry potentially correlated
secret bits.  Each vertex must recover its own bit from its incident shares,
and its complete view must reveal nothing beyond that bit.  We minimize the
number of source states used by the dealer, allowing arbitrary positive
source weights and finite share alphabets.

For $K_4$, let $D\subseteq\bits^4$ be the publicly known set of allowed
secret vectors, viewed as a table of four-bit rows.  For each row,
each source state specifies a labeling of all six edges, using the same
source-state probabilities for every row.  A vertex observes only its three
incident labels (Fig.~\ref{fig:k4-views}).  Its possible views for local bits
zero and one must be disjoint, so it can decode its bit.  Rows agreeing on
that bit must nevertheless produce identical weighted view distributions,
even when the other bits differ.  This is the privacy requirement: knowing
its own bit, a vertex cannot distinguish which compatible row of $D$ was
chosen.  We must meet these requirements at all four
vertices simultaneously: the label on each edge is shared by its two endpoints.

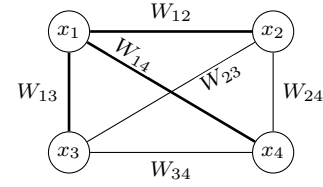
\begin{figure}[ht]
\centering
\begin{tikzpicture}[x=1cm,y=1cm,font=\footnotesize]
\node[draw,circle,fill=white,inner sep=2pt] (v1) at (0,1.6) {$x_1$};
\node[draw,circle,fill=white,inner sep=2pt] (v2) at (2.7,1.6) {$x_2$};
\node[draw,circle,fill=white,inner sep=2pt] (v3) at (0,0) {$x_3$};
\node[draw,circle,fill=white,inner sep=2pt] (v4) at (2.7,0) {$x_4$};
\draw[line width=.4pt] (v2) -- node[right] {$W_{24}$} (v4);
\draw[line width=.4pt] (v3) -- node[below] {$W_{34}$} (v4);
\draw[line width=.4pt] (v2) -- node[pos=.24,below,sloped,fill=white,inner sep=1pt] {$W_{23}$} (v3);
\draw[line width=1pt] (v1) -- node[above] {$W_{12}$} (v2);
\draw[line width=1pt] (v1) -- node[left] {$W_{13}$} (v3);
\draw[line width=1pt] (v1) -- node[pos=.24,above,sloped,fill=white,inner sep=1pt] {$W_{14}$} (v4);
\end{tikzpicture}
\caption{Six edge shares on $K_4$.  The three thick edges form vertex $1$'s
view $(W_{12},W_{13},W_{14})$; every vertex observes its own incident triple.}
\label{fig:k4-views}
\end{figure}

Earlier work studied randomness in multi-secret sharing through entropy
bounds~\cite{BlundoMasucci1999,Masucci2006} and connected correlated
multi-secret sharing to secure computation~\cite{Data2016}.
Anilkumar et al.~\cite{Anilkumar2023} classified every nonempty binary
three-secret sharing (3SS) domain on $K_3$.  We give the complete $K_4$
classification, with logarithmic support costs $0,1,\log_2 3,2$.
The main structural result is an exact one-bit characterization for arbitrary
graphs, derived directly from privacy and correctness without a linearity
assumption.  To prove it, we formulate a signed-subgraph criterion in terms
of cycle turn vectors.  A turn vector records the secret coordinates that
remain when edge-compatibility equations are XORed around a cycle; the
criterion requires the resulting parities to be constant on the domain.
Its $K_4$ specialization provides the first boundary of the classification.

The main text develops the one-bit mechanism, gives explicit three- and
four-state constructions, and follows the route from minimal obstructions
to the final classification.  Appendices supply the remaining structural
proofs, three-state reductions, and orbit calculations.  Exhaustive computation partitions the $896$ minimal
obstructions into six symmetry classes and excludes three states in four of
them; these searches are proof components of the classification.  All other
reductions and constructions have direct proofs.  The exact search and
an independent audit of the $K_4$ calculations are described in
Appendix~\ref{app:computation}.

\section{Model and Affine Relations}\label{sec:model}

The model below is the case of the CMSS framework of
Data et al.~\cite{Data2016} that protects each vertex
individually, written pointwise on the allowed domain $D$.
On $K_3$ it is the three-secret sharing setting of
Anilkumar et al.~\cite{Anilkumar2023}.

\begin{schemedefinition}
Let $G=(V,E)$ be a finite simple graph and let $D\subseteq\F_2^V$ be nonempty.
A scheme consists of a finite probability space $(R,\mu)$ with full support
($\mu(r)>0$ for every $r\in R$), whose distribution is independent of the
input $x\in D$, and share maps in arbitrary finite alphabets,
\[
 W_e:D\times R\to\mathcal W_e,\qquad e\in E.
\]
Vertex $v$ sees $\View_v(x,r)=(W_e(x,r):e\ni v)$.
Correctness requires
\begin{equation}
 \View_v(x,r)=\View_v(y,s)\ \Longrightarrow\ x_v=y_v
 \label{eq:correctness}
\end{equation}
for all $x,y\in D$ and $r,s\in R$.
For a source draw $\mathsf R\sim\mu$, perfect privacy requires
\begin{equation}
 \View_v(x,\mathsf R)\stackrel{d}{=}\View_v(y,\mathsf R)
 \quad\text{whenever }x_v=y_v,
 \label{eq:privacy}
\end{equation}
where $\stackrel{d}{=}$ denotes equality in distribution.
\end{schemedefinition}

Privacy concerns each complete vertex view, without an input distribution or
a coalition-privacy requirement.  Information implied by $x_v$ and $D$ is
allowed.  Each edge share inherits privacy by projection of either endpoint's
view: its distribution depends on the input only through either endpoint's
local secret.  Define
\begin{equation}
 m_G(D):=\min |R|,\qquad \rho_G(D):=\log_2 m_G(D),
 \label{eq:complexity}
\end{equation}
over correct, private schemes, with value $+\infty$ if none exists.
Thus $m_G(D)=1$ means that a deterministic scheme exists, and $m_G(D)\le2$
means that at most one random bit suffices: Lemma~\ref{lem:uniform-two}
shows that any two-state source may be made uniform.
The cost is logarithmic support, not minimum Shannon entropy.

Correlations can make an unmasked share private.  On $K_4$, for example, if
$x_1\xor x_2=c$ throughout $D$, with a fixed public bit $c$, assigning
$W_{12}=x_1$ lets vertices $1$ and $2$ recover $x_1$ and
$x_2=W_{12}\xor c$, respectively.  At either endpoint the share is already
determined by the local bit, so it reveals nothing beyond that bit.  This
handles these two vertices; a full scheme must also serve the others.
The affine relations below record such parity constraints on the domain.

For $x^0\in D$, let
\begin{equation}
 L(D):=\operatorname{span}\{x\xor x^0:x\in D\},\quad
 \Rel(D):=L(D)^\perp.
 \label{eq:relation-space}
\end{equation}
This is the direction space of the affine hull: changing the base point adds
another vector of $L(D)$ and leaves the span unchanged.  With $a\cdot x$
denoting the standard inner product over $\F_2$,
\begin{equation}
 a\in\Rel(D)\ \Longleftrightarrow\ a\cdot x\text{ is constant on }D.
 \label{eq:rel-affine}
\end{equation}
A vertex is \emph{active} if both local bits occur.  An edge $uv$ is a
\emph{pair-relation edge} if
\begin{equation}
 e_u+e_v\in\Rel(D),
 \label{eq:pair-relation}
\end{equation}
where $e_u,e_v$ are the corresponding standard unit vectors; equivalently,
$x_u\xor x_v$ is constant on $D$.  A vertex is \emph{pair-covered} if it is
incident to such an edge.

\section{The Exact One-Bit Theorem}\label{sec:onebit}

\begin{proposition}[Zero randomness]\label{prop:zero}
A deterministic scheme exists if and only if every active vertex is incident to
a pair-relation edge.
\end{proposition}
The proof is given in Appendix~\ref{app:one-bit-proofs}.

For a two-state source, Lemma~\ref{lem:uniform-two} permits replacing its
positive probabilities by $1/2,1/2$ without changing any share map;
correctness and privacy are preserved.  We therefore take the source
uniform on $\F_2$ in the following two-state analysis.

For each vertex $v$ and each occurring local bit $b$, choose one allowed
input $x^{v,b}\in D$ with $x^{v,b}_v=b$; this is the \emph{reference input}
for that vertex and bit.  Its views at source states $0,1$ fix the reference
ordering $A_v^b(s):=\View_v(x^{v,b},s)$ for $s\in\F_2$.

\begin{lemma}[Local switching form]\label{lem:local-switching}
For every vertex $v$ there is a \emph{switch bit}
$\alpha_v:D\to\F_2$ such that
\begin{equation}
  \View_v(x,r)=A_v^{x_v}(r\xor\alpha_v(x)).
  \label{eq:local-switching}
\end{equation}
If both local bits occur, the reference-view sets
$\{A_v^0(0),A_v^0(1)\}$ and $\{A_v^1(0),A_v^1(1)\}$ are disjoint.
\end{lemma}
\begin{proof}
With the uniform source, privacy makes the two views at $x$ the same
unordered pair, with repetitions allowed, as at its reference input.
Set $\alpha_v(x)=0$ when their ordering agrees with the reference and
$\alpha_v(x)=1$ when it is reversed.  If the two reference views coincide,
choose $\alpha_v(x)=0$.  A view shared by different local bits would
contradict correctness.
\end{proof}

The construction has a simple local mechanism.  Two incident shares of the
form $(U_v,U_v\xor x_v)$ reveal $x_v$ by XOR, while their distribution
depends only on $x_v$ when $U_v$ is uniform.  In the construction below,
the variables $U_v$ need not coincide across vertices: for each fixed input
they are XOR-shifts of the same source bit $R$, not additional independent
randomness.
The difficulty is making these local descriptions agree on every shared edge.

Consider one edge $e=uv$ that is not a pair-relation edge.
Appendix~\ref{app:one-bit-proofs} shows that endpoint privacy makes its
share either constant on $D\times R$ or, at every input, a bijection from
the two source states onto a common two-symbol set.  Call an edge of the
second kind \emph{regular}.  After labeling its two symbols by $0,1$,
the edge coordinate of the reference view at endpoint $w\in\{u,v\}$ has
form $s\xor f_{w,e}(b)$.  Since $b$ is binary, write
$f_{w,e}(b)=d_{w,e}\xor q_w(e)b$ with fixed bits $d_{w,e},q_w(e)$.
The coefficient $q_w(e)$ is the \emph{incidence sign} at $w$.
Both endpoints see the same edge share, so their local orderings must satisfy
\begin{equation}
 \alpha_u(x)\xor\alpha_v(x)
 =c_e\xor q_u(e)x_u\xor q_v(e)x_v,
 \label{eq:edge-compatibility}
\end{equation}
where $c_e=d_{u,e}\xor d_{v,e}$.  Thus the switch bits at adjacent vertices
cannot be chosen independently.

For example, if three regular edges form the cycle $1-2-3-1$, XORing their
compatibility equations cancels the switch bits:
\[
 (\alpha_1\xor\alpha_2)\xor(\alpha_2\xor\alpha_3)
 \xor(\alpha_3\xor\alpha_1)=0,
\]
where all switch bits are evaluated at the same input $x$.
At vertex $1$, the remaining coefficient of $x_1$ is
$q_1(12)\xor q_1(31)$, and similarly at the other vertices.
A secret bit survives exactly when the two incident cycle signs differ;
the XOR of the surviving bits equals the fixed constant
$c_{12}\xor c_{23}\xor c_{31}$ throughout $D$.
This is how an affine relation arises from a cycle.

We encode the same cancellation for arbitrary cycles as follows.
Let $H=(V,F)\subseteq G$ and fix incidence signs $q_v(e)\in\F_2$ for
$v\in e\in F$.  Define the \emph{binary cycle space} by
\begin{equation}
 Z_1(H;\F_2)=\left\{z\in\F_2^F:
 \bigoplus_{\substack{e\in F\\e\ni v}}z_e=0
 \text{ for every }v\in V\right\}.
 \label{eq:cycle-space}
\end{equation}
Here $z_e=1$ selects edge $e$.  The indicator of every simple cycle belongs
to this space, and every element is an XOR of simple-cycle indicators.
On a simple cycle, the \emph{turn vector} records the vertices whose two
cycle incidences have different signs, namely the secret bits that survive
the cancellation above.  For any $z\in Z_1(H;\F_2)$, define
\begin{equation}
  (\turn_q(z))_v:=\bigoplus_{\substack{e\in F\\e\ni v}}q_v(e)z_e.
  \label{eq:turn-map}
\end{equation}
For this fixed labeling $q$, the signed subgraph $(H,q)$ is
\emph{$D$-flat} if
\begin{equation}
 \turn_q(z)\in\Rel(D)\qquad\text{for every }z\in Z_1(H;\F_2).
 \label{eq:flatness}
\end{equation}
Flatness means that every parity constraint obtained by going around a cycle
already holds throughout the allowed domain.  For a simple cycle $C$, the
parity $\turn_q(C)\cdot x$ must have the same value for every $x\in D$;
this value may be either zero or one.  Equivalently, for any fixed $x^0\in D$,
\[
 \turn_q(C)\cdot(x\xor x^0)=0\qquad\text{for every }x\in D.
\]
As the construction below shows, this is what makes the vertex shifts
defined by path sums independent of the chosen path.

A vertex is \emph{switched} if both incidence signs occur among its selected
incident edges.

\begin{theorem}[Exact one-bit characterization]\label{thm:onebit}
For every finite graph $G$ and every nonempty binary domain $D$,
$m_G(D)\le2$ if and only if there is a signed subgraph $(H,q)$ such that
\begin{enumerate}[label=(\roman*),leftmargin=*]
\item $H$ contains no pair-relation edge;
\item $(H,q)$ is $D$-flat;
\item every active vertex is either incident in $G$ to a pair-relation edge or
is switched in $(H,q)$.
\end{enumerate}
\end{theorem}
We call $(H,q)$ a \emph{certificate}: these three combinatorial properties
suffice to construct a one-bit scheme explicitly, and conversely every
one-bit scheme yields such a pair.

\begin{proof}[Sufficiency]
Fix $x^0\in D$.  For $e=uv\in F$, set
\begin{equation}
 g_e(x)=q_u(e)(x_u\xor x_u^0)\xor q_v(e)(x_v\xor x_v^0).
 \label{eq:g-edge}
\end{equation}
Flatness gives
\[
 \bigoplus_{e\in F}z_eg_e(x)=\turn_q(z)\cdot(x\xor x^0)=0
 \qquad(z\in Z_1(H;\F_2)).
\]
Choose a root in each connected component of $H$.  Define $h_v(x)$ as the
XOR of the $g_e(x)$ along a path from the root to $v$, with value zero at
the root.  The XOR of two such paths belongs to the cycle space, so the
vanishing cycle sums make $h_v(x)$ independent of the chosen path.
In particular, for $e=uv\in F$,
\begin{equation}
 h_u(x)\xor h_v(x)=g_e(x).
 \label{eq:potential}
\end{equation}
For uniform $R\in\F_2$, the share on $e=uv\in F$ can therefore be written
at either endpoint as
\begin{equation}
\begin{aligned}
 W_e(x,R)&=R\xor h_u(x)\xor q_u(e)(x_u\xor x_u^0)\\
         &=R\xor h_v(x)\xor q_v(e)(x_v\xor x_v^0).
\end{aligned}
 \label{eq:selected-share}
\end{equation}
Indeed, the XOR of the two displayed expressions is
$h_u(x)\xor h_v(x)\xor g_e(x)=0$ by \eqref{eq:g-edge}--\eqref{eq:potential}.
Pair-relation edges carry one endpoint bit and all remaining edges carry
constants.  At $v$, write $U_v=R\xor h_v(x)$; the selected shares are
$(U_v\xor q_v(e)(x_v\xor x_v^0):e\ni v)$.
Since $U_v$ is uniform, this tuple's distribution depends only on $x_v$;
pair-relation shares are also determined by $x_v$, proving privacy of the
complete view.  Opposite incidence signs reveal $x_v\xor x_v^0$, and a
pair-relation edge reveals $x_v$ deterministically.  Every active vertex
therefore reconstructs, while an inactive vertex already has a fixed bit.
\end{proof}

\begin{proof}[Necessity (sketch)]
A deterministic scheme yields the empty selected subgraph by
Proposition~\ref{prop:zero}.  Otherwise uniformize the two-state source and
apply Lemma~\ref{lem:local-switching}.  Keep the regular edges introduced
above; Appendix~\ref{app:one-bit-proofs} derives their compatibility equation
\eqref{eq:edge-compatibility} from endpoint privacy.
Correctness forces every active vertex not covered by a pair relation to see
both incidence signs, hence to be switched.  XORing
\eqref{eq:edge-compatibility} over any $z\in Z_1(H;\F_2)$ cancels every
switch bit and gives
$\turn_q(z)\cdot x=\bigoplus_{e\in F:z_e=1}c_e$ throughout $D$.
Thus the signed subgraph $(H,q)$ is $D$-flat.
The full necessity argument is in Appendix~\ref{app:one-bit-proofs}.
\end{proof}

\section{The Exact \texorpdfstring{$K_4$}{K4} Classification}\label{sec:k4-onebit}

Let $\wt(a)$ denote the Hamming weight of $a$, the number of its nonzero
coordinates, and set
$E_{\rm even}:=\{a\in\F_2^4:\wt(a)\text{ is even}\}\subseteq\F_2^4$.
On $K_4$, the signed certificate reduces to a single affine relation.

\begin{theorem}[$K_4$ one-bit criterion]\label{thm:k4-onebit}
For every nonempty $D\subseteq\F_2^4$,
\begin{equation}
 m_{K_4}(D)\le2
 \quad\Longleftrightarrow\quad
 \Rel(D)\cap E_{\rm even}\ne\{0\}.
 \label{eq:k4-onebit}
\end{equation}
Equivalently, some pair parity or the total parity is constant on $D$.
\end{theorem}

\begin{proof}[Sufficiency]
Suppose first that a pair parity is constant.  Relabel so that $12$ is a
pair-relation edge, and sign the path $1-3-4-2$ oppositely at $3$ and $4$
(Fig.~\ref{fig:k4-certificates}(a)).  Delete every path edge that is itself a
pair-relation edge.  The remaining forest contains no such edge and is
$D$-flat automatically: its cycle space is $\{0\}$.
Vertices $1,2$ are pair-covered; each of $3,4$ either is pair-covered
by a deleted incident edge or retains both opposite signs and is switched.
Theorem~\ref{thm:onebit} therefore gives a scheme with at most two states.

If the total parity is constant and there is no pair-relation edge, sign a
Hamiltonian cycle oppositely at every vertex (Fig.~\ref{fig:k4-certificates}(b)).
Its cycle space has one nonzero element, whose turn vector is $1111\in\Rel(D)$.
Thus the cycle is $D$-flat and switches every vertex, so the same theorem
applies.  The case with a pair-relation edge was already handled above.
\end{proof}
The converse is proved in Appendix~\ref{app:k4-signed} using the signed-graph
alternatives in Lemma~\ref{lem:k4-turns}.

\begin{figure}[ht]
\centering
\begin{tikzpicture}[x=1cm,y=1cm,font=\footnotesize,
  vertex/.style={draw,circle,fill=white,inner sep=1.8pt},
  incidence/.style={fill=white,inner sep=.8pt,font=\scriptsize}]
\begin{scope}
\node[vertex] (a1) at (0,1.5) {$1$};
\node[vertex] (a2) at (2.2,1.5) {$2$};
\node[vertex] (a3) at (0,0) {$3$};
\node[vertex] (a4) at (2.2,0) {$4$};
\draw[dashed] (a1) -- node[above,font=\scriptsize] {pair relation} (a2);
\draw[thick] (a1) -- node[incidence,pos=.19,right] {$0$}
  node[incidence,pos=.81,right] {$0$} (a3);
\draw[thick] (a3) -- node[incidence,pos=.16,above] {$1$}
  node[incidence,pos=.84,above] {$0$} (a4);
\draw[thick] (a4) -- node[incidence,pos=.19,left] {$1$}
  node[incidence,pos=.81,left] {$0$} (a2);
\node at (1.1,-.48) {(a) Pair parity};
\end{scope}
\begin{scope}[xshift=3.7cm]
\node[vertex] (b1) at (0,1.5) {$1$};
\node[vertex] (b2) at (2.2,1.5) {$2$};
\node[vertex] (b3) at (0,0) {$3$};
\node[vertex] (b4) at (2.2,0) {$4$};
\draw[thick] (b1) -- node[incidence,pos=.19,right] {$0$}
  node[incidence,pos=.81,right] {$0$} (b3);
\draw[thick] (b3) -- node[incidence,pos=.16,above] {$1$}
  node[incidence,pos=.84,above] {$0$} (b4);
\draw[thick] (b4) -- node[incidence,pos=.19,left] {$1$}
  node[incidence,pos=.81,left] {$0$} (b2);
\draw[thick] (b2) -- node[incidence,pos=.16,below] {$1$}
  node[incidence,pos=.84,below] {$1$} (b1);
\node at (1.1,-.48) {(b) Total parity};
\end{scope}
\end{tikzpicture}
\caption{One-bit certificates for $K_4$.  Labels $0,1$ belong to incidences,
not whole edges.  In (a), dashed edge $12$ gives pair coverage and lies outside
$H$; any pair-relation edges on the solid path are also deleted from $H$,
providing pair coverage at their endpoints.  In (b), no pair-relation edge
exists: every vertex has opposite incidence signs, so the cycle turn is $1111$.}
\label{fig:k4-certificates}
\end{figure}
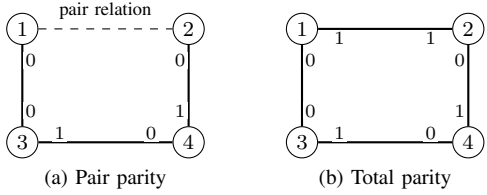

\subsection{Minimal obstructions}

To organize the remaining domains, let
\begin{equation}
 \cG:=\F_2^4\rtimes S_4
 \label{eq:gamma4}
\end{equation}
act by independent bit negations and coordinate permutations.  Write
$D\sim_{\cG}D'$ for equivalence under this action, which preserves $m_{K_4}$,
and define
\begin{equation}
 T_6:=\{0000,0111,1001,1010\},
 \label{eq:T6}
\end{equation}
\begin{equation}
 C_5:=T_6\cup\{1100\}.
 \label{eq:C5}
\end{equation}

Restriction preserves a scheme: if $D'\subseteq D$ is nonempty, then
$m_{K_4}(D')\le m_{K_4}(D)$ (Lemma~\ref{lem:heredity}).  Thus a small
subdomain can witness that a larger domain needs more than two states.
Call $\{x^0,x^1,x^2,x^3\}$ an \emph{odd affine tetrahedron} if the three
differences $x^i\xor x^0$ ($i=1,2,3$) are linearly independent and the
unique nonzero normal to their span has odd Hamming weight.
Theorem~\ref{thm:odd-tetra} shows that these are exactly the
inclusion-minimal domains with $m>2$.  Write $\cO_1,\ldots,\cO_6$ for
their six $\cG$-orbits, indexed by the representatives in
Table~\ref{tab:orbit-summary}, which also gives their exact support values.
The geometric proof and full orbit enumeration are in
Appendix~\ref{app:k4-signed}.

\begin{table}[ht]
\caption{Minimal one-bit obstructions and exact support values.
The listed steps are complemented by the four-state construction below for
$\cO_1$--$\cO_5$ and the lower bound $m>2$ for $\cO_6$.}
\label{tab:orbit-summary}
\centering
\footnotesize
\setlength{\tabcolsep}{3pt}
\begin{tabular}{@{}c l c l@{}}
\toprule
Orbit & Representative & $m_{K_4}$ & Proof step\\
\midrule
$\cO_1$ & $\{0000,0001,0010,0100\}$ & $4$ & direct lower bound\\
$\cO_2$ & $\{0000,0001,0011,0100\}$ & $4$ & exhaustive search\\
$\cO_3$ & $\{0000,0010,0011,0101\}$ & $4$ & exhaustive search\\
$\cO_4$ & $\{0000,0011,0101,1000\}$ & $4$ & exhaustive search\\
$\cO_5$ & $\{0000,0101,1000,1011\}$ & $4$ & exhaustive search\\
$\cO_6$ & $\{0000,0111,1001,1010\}=T_6$ & $3$ & $\Z_3$ scheme\\
\bottomrule
\end{tabular}
\end{table}

The three-state test must cover arbitrary source weights and finite share
alphabets.  Set
\[
 \cU=(1/3,1/3,1/3),\qquad \cQ=(1/2,1/4,1/4).
\]
For a domain with $m>2$, the weight-collapse proposition leaves only these
two distributions, up to permutation.  The exact normal form then replaces
arbitrary alphabets by finitely many local-view labels and the edge equalities
they force.  A configuration is accepted when these equalities preserve the
distinction between every pair of different local-view labels.
Every scheme induces an accepted configuration, and every accepted
configuration yields a scheme, so the search is complete without a linearity
assumption.  The reduction and normal form are proved in
Appendix~\ref{app:three-state}
(Proposition~\ref{prop:weight-collapse}, Theorem~\ref{thm:three-normal});
implementation details are in Appendix~\ref{app:computation}.

\subsection{Explicit three- and four-state schemes}

The five-point domain $C_5$ has a useful arithmetic description:
\begin{equation}
 C_5=\{x\in\bits^4:x_1\equiv x_2+x_3+x_4\pmod 3\}.
 \label{eq:C5-congruence}
\end{equation}
Here the bits $x_i\in\{0,1\}$ are viewed as integers and the addition in
\eqref{eq:C5-congruence} is modulo $3$, not XOR.
Indeed, for $x\in C_5$ with $x_1=0$, the defining congruence forces the
integer sum $x_2+x_3+x_4$ to be $0$ or $3$, giving two points.
For $x\in C_5$ with $x_1=1$, it forces the sum to be $1$, giving the other
three points.
This congruence closes a masked sum around the cycle $1-2-3-4-1$:
successive differences recover the local bits.  Only four edges carry
nonconstant shares; the diagonals $13$ and $24$ are constant
(Fig.~\ref{fig:c5-cycle}).

The roles of $T_6$ and $C_5$ are distinct.  The four-point set
$T_6=C_5\setminus\{1100\}$ is a minimal one-bit obstruction, which supplies
the lower bound and the starting point of the classification argument.
The congruence supplies a three-state construction on the larger set $C_5$,
and therefore on $T_6$ as well.

\begin{figure}[!ht]
\centering
\begin{tikzpicture}[x=1cm,y=1cm,font=\footnotesize,
  vertex/.style={draw=blue!45!black,circle,fill=white,
    minimum size=16pt,inner sep=1.4pt,line width=.65pt},
  cycleedge/.style={draw=blue!45!black,line width=.95pt},
  zeroedge/.style={draw=black!55,densely dashed,line width=.5pt},
  zerolabel/.style={fill=white,inner sep=1.2pt,text=black!70}]
\node[vertex] (z1) at (0,1.9) {$1$};
\node[vertex] (z2) at (3.8,1.9) {$2$};
\node[vertex] (z3) at (3.8,0) {$3$};
\node[vertex] (z4) at (0,0) {$4$};
\draw[zeroedge] (z1) -- node[zerolabel,pos=.27,below,sloped] {$0$} (z3);
\draw[zeroedge] (z2) -- node[zerolabel,pos=.27,below,sloped] {$0$} (z4);
\draw[cycleedge] (z1) -- node[above=3pt] {$R$} (z2);
\draw[cycleedge] (z2) -- node[right=4pt] {$R-x_2$} (z3);
\draw[cycleedge] (z4) -- node[below=3pt] {$R-x_2-x_3$} (z3);
\draw[cycleedge] (z1) -- node[left=4pt] {$R-x_1$} (z4);
\end{tikzpicture}
\caption{The three-state scheme on $C_5$, with uniform $R\in\Z_3$.
Cycle edges carry the displayed shares; dashed diagonals carry zero.
All displayed share arithmetic is in $\Z_3$.  At vertex $4$, the congruence
\eqref{eq:C5-congruence} gives
$W_{34}-W_{14}=x_1-x_2-x_3=x_4$.}
\label{fig:c5-cycle}
\end{figure}
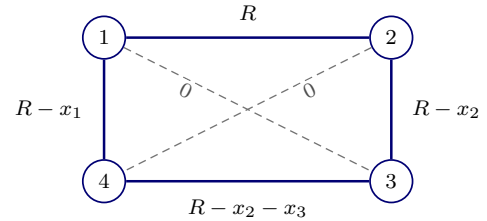

\begin{theorem}[Uniform $\Z_3$ scheme]\label{thm:C5-scheme}
Every binary domain contained in $C_5$ has a three-state scheme.  With
$R\sim\Unif(\Z_3)$, set, in $\Z_3$,
\begin{align}
 W_{12}&=R,            &W_{13}&=0,          &W_{14}&=R-x_1,\nonumber\\
 W_{23}&=R-x_2,        &W_{24}&=0,          &W_{34}&=R-x_2-x_3.
 \label{eq:C5-shares}
\end{align}
\end{theorem}
\begin{proof}
Parties $1$ and $2$ see $(R,0,R-x_1)$ and $(R,R-x_2,0)$ and recover their
bits by subtraction.  With $U=R-x_2$, party $3$ sees $(0,U,U-x_3)$; with
$U=R-x_1$ and $x_1-x_2-x_3\equiv x_4\pmod3$, party $4$ sees
$(U,0,U+x_4)$.  Thus each party reconstructs, while uniformity of $U$ makes
its complete-view distribution depend only on its local bit.  Restricting the
input domain preserves the scheme.
\end{proof}

A direct row reduction of the difference vectors based at $0000$ gives
\begin{equation}
 \Rel(T_6)=\operatorname{span}\{1011\},
 \qquad
 \Rel(C_5)=\{0\}.
 \label{eq:T6C5-rel}
\end{equation}
Thus neither relation space contains a nonzero even vector.
Theorem~\ref{thm:k4-onebit} gives $m>2$, while
Theorem~\ref{thm:C5-scheme} gives
$m\le3$.
Therefore
\begin{equation}
  m_{K_4}(T_6)=m_{K_4}(C_5)=3.
  \label{eq:T6C5-values}
\end{equation}

For the universal upper bound on $K_4$, split its six edges into the paths
$2-1-3-4$ and $1-4-2-3$, masked by independent uniform bits $R$ and $S$.
Along each path, start with its mask and XOR in the secret at each internal
vertex to obtain the next share.  Every vertex is internal to exactly one
path and an endpoint of the other.  Its internal pair recovers its bit;
for each fixed input, the remaining share is a uniformly masked bit
independent of that pair.  This explains the view
$(U_v,U_v\xor x_v,Z_v)$ used in the proof below;
Fig.~\ref{fig:k4-two-paths} displays the two paths.

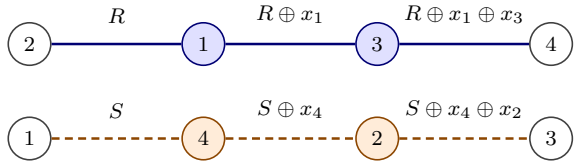
\begin{figure}[!ht]
\centering
\begin{tikzpicture}[x=1cm,y=1cm,font=\footnotesize,
  vertex/.style={draw=black!75,circle,fill=white,
    minimum size=16pt,inner sep=1.4pt,line width=.6pt},
  redge/.style={draw=blue!45!black,line width=.95pt},
  sedge/.style={draw=orange!55!black,line width=.95pt,
    dash pattern=on 3pt off 2pt},
  rinternal/.style={vertex,draw=blue!45!black,fill=blue!12},
  sinternal/.style={vertex,draw=orange!55!black,fill=orange!15}]
\node[vertex] (r2) at (0,1.25) {$2$};
\node[rinternal] (r1) at (2.3,1.25) {$1$};
\node[rinternal] (r3) at (4.6,1.25) {$3$};
\node[vertex] (r4) at (6.9,1.25) {$4$};
\draw[redge] (r2) -- node[above=4pt] {$R$} (r1);
\draw[redge] (r1) -- node[above=4pt] {$R\xor x_1$} (r3);
\draw[redge] (r3) -- node[above=4pt] {$R\xor x_1\xor x_3$} (r4);
\node[vertex] (s1) at (0,0) {$1$};
\node[sinternal] (s4) at (2.3,0) {$4$};
\node[sinternal] (s2) at (4.6,0) {$2$};
\node[vertex] (s3) at (6.9,0) {$3$};
\draw[sedge] (s1) -- node[above=4pt] {$S$} (s4);
\draw[sedge] (s4) -- node[above=4pt] {$S\xor x_4$} (s2);
\draw[sedge] (s2) -- node[above=4pt] {$S\xor x_4\xor x_2$} (s3);
\end{tikzpicture}
\caption{The four-state scheme: two edge-disjoint paths of the same $K_4$,
with independent uniform bits $R,S$.  Repeated numbers denote the same
vertex.  Filled vertices are internal to their path: each vertex reconstructs
on one path and receives an independently masked share from the other.}
\label{fig:k4-two-paths}
\end{figure}
\FloatBarrier

\begin{proposition}[Four-state upper bound]\label{prop:k4-upper}
For every nonempty $D\subseteq\F_2^4$, one has $m_{K_4}(D)\le4$.
\end{proposition}
\begin{proof}
Let $R,S$ be independent uniform bits.  On the six edges of $K_4$, set
\begin{align}
 W_{12}&=R,& W_{13}&=R\xor x_1,&
 W_{34}&=R\xor x_1\xor x_3,\nonumber\\
 W_{14}&=S,& W_{24}&=S\xor x_4,&
 W_{23}&=S\xor x_4\xor x_2.
 \label{eq:k4-four-state}
\end{align}
Correctness follows from
\[
\begin{aligned}
 x_1&=W_{12}\xor W_{13},& x_2&=W_{23}\xor W_{24},\\
 x_3&=W_{13}\xor W_{34},& x_4&=W_{14}\xor W_{24}.
\end{aligned}
\]
Up to a fixed ordering of its incident edges, each vertex $v$ sees
$(U_v,U_v\xor x_v,Z_v)$.  For vertices $1,2,3,4$, respectively, take
\[
\begin{aligned}
 (U_1,Z_1)&=(R,S),\\
 (U_2,Z_2)&=(S\xor x_4,R),\\
 (U_3,Z_3)&=(R\xor x_1,S\xor x_4\xor x_2),\\
 (U_4,Z_4)&=(S,R\xor x_1\xor x_3).
\end{aligned}
\]
For every fixed input these are pairs of independent uniform bits, so the
complete view distribution depends only on $x_v$.  The source $(R,S)$ has
four states, and the same formulas work on every subdomain $D$.
\end{proof}

\subsection{Completing the classification}

These ingredients give the final step: among the six obstruction orbits only
$\cO_6$, represented by $T_6$, admits three states, and its only additional
input compatible with three states is $1100$.  We now make this closure
argument explicit.

\begin{theorem}[Exact binary $K_4$ classification]\label{thm:full-k4}
For every nonempty $D\subseteq\F_2^4$,
\begin{equation}
 m_{K_4}(D)=
 \begin{cases}
 1, & \text{if every active vertex is pair-covered},\\
 2, & \text{if not, and }\Rel(D)\cap E_{\rm even}\ne\{0\},\\
 3, & \text{if }D\sim_{\cG}T_6\text{ or }D\sim_{\cG}C_5,\\
 4, & \text{otherwise}.
 \end{cases}
 \label{eq:full-classification}
\end{equation}
Consequently,
\begin{equation}
  \rho_{K_4}(D)\in\{0,1,\log_2 3,2\},
  \label{eq:spectrum}
\end{equation}
and the value $\log_2 3$ occurs exactly in the orbits of $T_6$ and $C_5$.
\end{theorem}
\begin{proof}
Proposition~\ref{prop:zero} gives the first case, and
Theorem~\ref{thm:k4-onebit} gives the second after excluding the first.
Suppose now that $m_{K_4}(D)>2$ but $m_{K_4}(D)\le3$.
By Theorem~\ref{thm:odd-tetra}, $D$ contains a four-point minimal obstruction
$T$.  Heredity and the values in Table~\ref{tab:orbit-summary}
(proved in Theorem~\ref{thm:orbit-values}) force $T\in\cO_6$.
Apply a symmetry to normalize $T=T_6$.

For every $p\notin T_6\cup\{1100\}$, the extension check in
Table~\ref{tab:extension} in Appendix~\ref{app:k4-signed} exhibits an
$\cO_1$--$\cO_5$ obstruction inside $T_6\cup\{p\}$.  Its minimum support
is four, so heredity forbids $p\in D$.  Therefore
$T_6\subseteq D\subseteq C_5$, giving exactly $D=T_6$ or $D=C_5$.
Conversely, the construction in Theorem~\ref{thm:C5-scheme} and the
relation spaces in \eqref{eq:T6C5-rel} give minimum support three for
both sets.  Thus precisely these two orbits have minimum support three.
Every remaining domain needs at least four states, and
Proposition~\ref{prop:k4-upper} supplies a four-state scheme for every $K_4$
domain.
\end{proof}

\section{Conclusion}

Affine relations determine when one random bit suffices on $K_4$, but not
the full cost: $C_5$ and $\F_2^4$ share their affine hull yet need three
and four source states, respectively.
The classification also limits a known lower-bound route for secure
computation: a reduction to three-secret sharing yields a three-bit
lower bound for 1-private three-party $\mathsf{AND}$~\cite{Anilkumar2024}.
However, every binary $K_4$ domain admits a two-random-bit scheme,
so a reduction retaining only these sharing constraints cannot establish
a larger lower bound.
Natural next steps are randomness criteria on larger graphs and
reductions that retain the constraints of interaction.

\section*{Acknowledgments}

The author thanks Akin \"Unal for suggesting this problem and for detailed
comments on the manuscript, and Krzysztof Pietrzak for helpful discussions.

The author used OpenAI's ChatGPT/Codex for language editing, \LaTeX\ formatting,
and consistency checking.  All AI-assisted suggestions were reviewed by the
author, who remains responsible for the mathematical content, citations,
computational claims, and conclusions.

\clearpage
\appendices

\section{Proofs for the Exact One-Bit Theorem}\label{app:one-bit-proofs}

\subsection{Basic reductions}

\begin{lemma}[Heredity]\label{lem:heredity}
If $D'\subseteq D$ is nonempty, then $m_G(D')\le m_G(D)$.
\end{lemma}
\begin{proof}[Proof of Lemma~\ref{lem:heredity}]
Restrict every share map to $D'\times R$.  Both
\eqref{eq:correctness} and \eqref{eq:privacy} are restrictions of conditions
already valid on $D$.
\end{proof}

\begin{proof}[Proof of Proposition~\ref{prop:zero}]
In a deterministic scheme, privacy makes every edge share $W_{uv}$ a function
of $x_v$ at endpoint $v$ and a function of $x_u$ at endpoint $u$.  If $v$ is
active, correctness forces some incident share to be nonconstant in $x_v$.
Because the secrets are binary, that share determines $x_v$ from $x_u$, and
$x_u\xor x_v$ is constant on $D$.

Conversely, orient each pair-relation edge $u\to v$ and set $W_{uv}=x_u$.
At $v$ this equals $x_v\xor c_{uv}$ for the public constant
$c_{uv}=x_u\xor x_v$.  Set constants on all remaining edges.  Every active
vertex reconstructs from one pair edge, and every incident coordinate is a
function of the local secret.
\end{proof}

The next observation removes a possible source-weight pathology at support two.

\begin{lemma}[Uniformization]\label{lem:uniform-two}
If a correct and private scheme has source support $R=\{0,1\}$ with arbitrary
positive probabilities, then replacing the source by the uniform distribution
on the same two states preserves correctness and privacy.
\end{lemma}
\begin{proof}
Correctness depends only on the supported states.  For privacy, fix $v$ and
$x,y$ with $x_v=y_v$, and write
\[
\begin{aligned}
 a&=\View_v(x,0),& b&=\View_v(x,1),\\
 c&=\View_v(y,0),& d&=\View_v(y,1).
\end{aligned}
\]
Let the source probabilities be $p,1-p$.  If $p=1/2$, nothing changes.
Otherwise, if $a=b$, the view at $x$ is constant; equal distributions and
positive source probabilities force $c=d=a$.  If $a\ne b$, the view at $x$
has two distinct values of unequal probabilities $p,1-p$.
The equal distribution at $y$ identifies which value comes from each source
state, so $a=c$ and $b=d$.  In both cases, assigning probability $1/2$ to
each state preserves equality of the two view distributions.
\end{proof}

\subsection{Necessity}

Lemma~\ref{lem:local-switching} and the sufficiency direction of
Theorem~\ref{thm:onebit} are proved in Section~\ref{sec:onebit}.
We now supply the remaining argument using the graph of occurring endpoint
bits.

\begin{proof}[Completion of the proof of Theorem~\ref{thm:onebit}]
The deterministic case is Proposition~\ref{prop:zero}, with $H$ empty.
Otherwise Lemma~\ref{lem:uniform-two} makes the source uniform; fix the local
representations \eqref{eq:local-switching}.

For a non-pair edge $e=uv$, let $S_u,S_v$ be the occurring endpoint bits and
form the bipartite support graph $B_e$ on $S_u\sqcup S_v$, joining $(b,c)$ when
some $x\in D$ has $(x_u,x_v)=(b,c)$.  This graph is connected.  Indeed, this is
immediate when exactly one endpoint is active; if both are active, a
disconnected support with no isolated vertex is one of the two perfect
matchings and makes $x_u\xor x_v$ constant; and if neither is active, $e$ is a
pair relation.

Let $\nu_x$ be the distribution of $W_e(x,\mathsf R)$.  Endpoint privacy
propagates equality of these distributions along $B_e$, so all $\nu_x$ equal
one distribution $\nu_e$.  Thus either
every $W_e(x,\cdot)$ is the same constant, or all these maps are bijections onto
a common two-symbol set.  For a nonconstant edge, choose a binary labeling
$\lambda_e$ of that set.  For $w\in\{u,v\}$ and $b\in S_w$, the $e$-coordinate
of $A_w^b(s)$ is also bijective in $s$, hence
\[
 \lambda_e\bigl((A_w^b(s))_e\bigr)=s\xor f_{w,e}(b)
\]
for some $f_{w,e}:S_w\to\F_2$.  Equality of the endpoint descriptions of the edge
share gives
\[
 \alpha_u(x)\xor\alpha_v(x)=f_{u,e}(x_u)\xor f_{v,e}(x_v).
\]
Every function on an occurring binary coordinate is affine: write
$f_{w,e}(b)=d_{w,e}\xor q_w(e)b$, taking $q_w(e)=0$ when $w$ is inactive.
With $c_e=d_{u,e}\xor d_{v,e}$, the endpoint identity becomes
\eqref{eq:edge-compatibility}.  The common two-symbol labeling thus reduces
arbitrary edge alphabets to a compatibility condition on the switch bits.

Call the nonconstant non-pair edges \emph{regular}, let $H$ consist of them,
and use the coefficients $q_w(e)$ as incidence signs.  If an active vertex $v$
with no incident pair relation were not switched, all its regular incidences
would have one sign $q$ (take $q=0$ if there are none).  The preceding display
gives the same coordinate in $A_v^0(0)$ and $A_v^1(q)$ on every regular edge;
every other incident edge is constant.  The two views would coincide,
contradicting Lemma~\ref{lem:local-switching}.  Thus condition~(iii) holds.

Finally XOR \eqref{eq:edge-compatibility} over the edges selected by any
$z\in Z_1(H;\F_2)$.  Every switch bit $\alpha_v(x)$ occurs an even number of times,
so
\[
 \turn_q(z)\cdot x=\bigoplus_{e:z_e=1}c_e
\]
is constant on $D$.  Equation~\eqref{eq:rel-affine} gives
$\turn_q(z)\in\Rel(D)$.  Hence $(H,q)$ is $D$-flat, contains no pair-relation
edge, and satisfies all conditions of Theorem~\ref{thm:onebit}.
\end{proof}

\section{The \texorpdfstring{$K_4$}{K4} Criterion, Obstructions, and Extensions}\label{app:k4-signed}

We supply the finite graph argument used in Theorem~\ref{thm:k4-onebit}.
For a simple cycle $C$ in a signed selected subgraph, write $\turn_q(C)$ for
the turn vector of its edge set.

\begin{lemma}[Signed $K_4$ alternatives]\label{lem:k4-turns}
Let $H\subseteq K_4$ be signed.
\begin{enumerate}[label=(\roman*),leftmargin=*]
\item If all four vertices are switched, some simple cycle $C$ satisfies
$\wt(\turn_q(C))\in\{2,4\}$.
\item If three designated active vertices are switched and $d$ is the fourth
vertex, some simple cycle $C$ satisfies
$\turn_q(C)\notin\{0,e_d\}$.
\end{enumerate}
\end{lemma}
\begin{proof}
For (i), switchedness gives minimum selected degree two, hence
$|E(H)|\in\{4,5,6\}$.  If $|E(H)|=4$, every vertex has selected degree two,
so $H$ is a four-cycle $C$.  Since every vertex is switched, the two cycle
incidences at each vertex have different signs.  Therefore
$\turn_q(C)=(1,1,1,1)$, which has weight four.
If $|E(H)|=5$, relabel so that $34$ is missing.  Let $T_3$ and $T_4$ be the
triangles with vertex sequences $1-2-3-1$ and $1-2-4-1$, respectively.
The XOR of their edge indicators is the cycle $C$ with sequence
$1-3-2-4-1$, so linearity gives
$\turn_q(C)=\turn_q(T_3)\xor\turn_q(T_4)$.
Vertices $3,4$ turn on $C$, so its turn weight is at least two; weights two
and four already give the claim.  If its turn weight is three, the XOR
identity makes the two triangle turn vectors have opposite weight parity.
Both are nonzero: vertex $3$ has selected degree two and is switched, so
$T_3$ turns at $3$; likewise $T_4$ turns at $4$.
The even-weight vector is nonzero and supported on a triangle, hence has
weight two.  If $H=K_4$, each
vertex turns on exactly two of the three Hamiltonian cycles.  Their turn weights
sum to eight, impossible if all three lie in $\{0,1,3\}$.

For (ii), suppose every simple cycle has turn vector in $\{0,e_d\}$.
Decompose the graph into blocks: each bridge is a two-vertex block, and the
other nontrivial blocks are maximal 2-connected subgraphs.  Distinct blocks
can meet only at cut vertices.  At an active vertex, any two incident edges
in the same 2-connected block lie on a common cycle and therefore have the
same sign.  A switched active vertex
must then meet distinct blocks and is a cut vertex.  This would make all three
active vertices cut vertices, whereas a graph on four vertices has at most two:
each nontrivial component has two non-cut vertices, obtained as leaves of a
spanning tree, and isolated vertices are not cut vertices.
\end{proof}

\begin{proof}[Necessity in Theorem~\ref{thm:k4-onebit}]
Sufficiency was proved in the main text.  Let $A$ be the set of active coordinates and apply
Theorem~\ref{thm:onebit}.  If $|A|\le2$, two coordinates are inactive, and their
XOR is a constant weight-two relation.  Suppose $|A|=3$ and let $d$ be inactive,
so $e_d\in\Rel(D)$.  A pair-relation edge already gives weight two.  Otherwise
all three active vertices are switched in a flat certificate.  By
Lemma~\ref{lem:k4-turns}(ii), a simple cycle has turn vector
$a\notin\{0,e_d\}$.  Flatness gives $a\in\Rel(D)$, and also
$a\xor e_d\in\Rel(D)$.  Exactly one of these two nonzero vectors has even weight,
which in four coordinates is two or four.

Finally, if all four coordinates are active, either a pair-relation edge gives
weight two or all four vertices are switched.  Lemma~\ref{lem:k4-turns}(i) and
flatness then give a relation of weight two or four.
\end{proof}

% Queue the wide enumeration table early for bottom placement;
% the page break before the extension subsection keeps the reading order.
\begin{table*}[!b]
\caption{Full enumeration data for Appendix~\ref{app:orbit-enumeration}:
the six $\cG$-orbits of minimal one-bit obstructions.  A
representative is one domain chosen from each $\cG$-orbit.  The last column
records the exact support value and its justification.}
\label{tab:orbits}
\centering
\scriptsize
\resizebox{\textwidth}{!}{%
\begin{tabular}{c l l c c c l}
\toprule
Orbit & Representative & nonzero normal & active coordinates & distance spectrum & orbit size & support value / justification\\
\midrule
$\cO_1$ & $\{0000,0001,0010,0100\}$ & $1000$ & $2,3,4$ & $\{1^3,2^3\}$ & $64$ & $m=4$, direct proof\\
$\cO_2$ & $\{0000,0001,0011,0100\}$ & $1000$ & $2,3,4$ & $\{1^3,2^2,3\}$ & $192$ & $m=4$, exhaustive\\
$\cO_3$ & $\{0000,0010,0011,0101\}$ & $1000$ & $2,3,4$ & $\{1^2,2^3,3\}$ & $192$ & $m=4$, exhaustive\\
$\cO_4$ & $\{0000,0011,0101,1000\}$ & $0111$ & $1,2,3,4$ & $\{1,2^3,3^2\}$ & $192$ & $m=4$, exhaustive\\
$\cO_5$ & $\{0000,0101,1000,1011\}$ & $0111$ & $1,2,3,4$ & $\{1,2^2,3^3\}$ & $192$ & $m=4$, exhaustive\\
$\cO_6$ & $T_6=\{0000,0111,1001,1010\}$ & $1011$ & $1,2,3,4$ & $\{2^3,3^3\}$ & $64$ & $m=3$, explicit uniform-$\Z_3$ scheme\\
\bottomrule
\end{tabular}%
}
\end{table*}

\subsection{Odd affine tetrahedra}

We prove the geometric characterization of minimal obstructions stated
in Section~\ref{sec:k4-onebit}.

\begin{theorem}[Minimal one-bit obstructions]\label{thm:odd-tetra}
A nonempty domain $D\subseteq\F_2^4$ satisfies $m_{K_4}(D)>2$ if and only if it contains
an odd affine tetrahedron.  Consequently, the inclusion-minimal one-bit-negative
domains are exactly the odd affine tetrahedra, and every minimal obstruction has
four points.
\end{theorem}
\begin{proof}
Theorem~\ref{thm:k4-onebit} says that negativity is equivalent to
\begin{equation}
 L(D)^\perp\cap E_{\rm even}=\{0\}.
 \label{eq:negative-even}
\end{equation}
Fix $x^0\in D$.  The evaluation map
$a\mapsto(a\cdot(x\xor x^0))_{x\in D\setminus\{x^0\}}$ is therefore injective on
$E_{\rm even}$.  Three of its coordinates retain rank three; call the
corresponding points $x^1,x^2,x^3$.  Their differences span a
three-dimensional space $U$, since otherwise
$\dim(E_{\rm even}\cap U^\perp)\ge1$.  The unique nonzero normal to $U$ is thus
odd, so $\{x^0,x^1,x^2,x^3\}$ is an odd affine tetrahedron.

Conversely, such a tetrahedron has a one-dimensional relation space generated
by an odd vector, hence no nonzero even relation and is negative by
Theorem~\ref{thm:k4-onebit}.  Heredity handles every domain containing it.
Every proper subset has $\dim L\le2$, so
$\dim(L^\perp\cap E_{\rm even})\ge1$ and is one-bit-positive, proving
minimality.
\end{proof}

\subsection{Full six-orbit enumeration}\label{app:orbit-enumeration}

There are $896$ odd affine tetrahedra.  Indeed, there are eight odd nonzero
normals, two affine hyperplanes per normal, and each affine copy of $\F_2^3$
contains
\[
  \frac{8\,|\mathrm{GL}(3,2)|}{4!}
  =\frac{8\cdot168}{24}=56
\]
unordered affine bases.  Exact orbit closure under $\cG$ partitions these
$896$ sets into the six classes in Table~\ref{tab:orbits}.  The \emph{distance
spectrum} of a four-point domain is the multiset of its six pairwise Hamming
distances; for example, $\{1^3,2^3\}$ denotes three pairs at distance one and
three pairs at distance two.  The six spectra are distinct, so the rows cannot
merge; the orbit sizes sum to $896$.
This orbit partition is a finite computer enumeration, independently reproduced
by the audit code described in Appendix~\ref{app:computation}.

\FloatBarrier
% A bottom two-column float must be shipped out before subsection C starts.
\clearpage
\subsection{The Unique Support-Three Extension of \texorpdfstring{$T_6$}{T6}}\label{app:extension}

We now use the six classes in Table~\ref{tab:orbits} to test extensions of
$T_6=\{0000,0111,1001,1010\}$.  Table~\ref{tab:extension} checks every cube
point outside $T_6$.  For each point $p\ne1100$ outside $T_6$, it specifies one point
$x^\ast\in T_6$ such that $(T_6\setminus\{x^\ast\})\cup\{p\}$ is a forbidden minimal
obstruction.  The orbit is identified by its Hamming-distance spectrum, which
is a complete invariant among the six rows of Table~\ref{tab:orbits}.
The exclusion of three-state extensions also uses the support-four lower
bounds for $\cO_1$--$\cO_5$ in Theorem~\ref{thm:orbit-values}, proved in
Appendices~\ref{app:three-state}--\ref{app:computation}; orbit membership
alone does not supply those lower bounds.  The admissible point $1100$ is
certified by the $C_5$ construction in Theorem~\ref{thm:C5-scheme}.

\begin{table}[ht]
\caption{Extension check for $T_6$.  The middle column gives the point of
$T_6$ omitted from the forbidden four-set.}
\label{tab:extension}
\centering
\scriptsize
\begin{tabular}{c c c}
\toprule
added $p$ & omitted point & orbit of $(T_6\setminus\{x^\ast\})\cup\{p\}$\\
\midrule
$0001$ & $1001$ & $\cO_5$\\
$0010$ & $1010$ & $\cO_5$\\
$0011$ & $1010$ & $\cO_4$\\
$0100$ & $1010$ & $\cO_5$\\
$0101$ & $1010$ & $\cO_4$\\
$0110$ & $1001$ & $\cO_4$\\
$1000$ & $0111$ & $\cO_1$\\
$1011$ & $1010$ & $\cO_5$\\
$1100$ & -- & $m=3$\\
$1101$ & $1010$ & $\cO_5$\\
$1110$ & $1001$ & $\cO_5$\\
$1111$ & $0000$ & $\cO_4$\\
\bottomrule
\end{tabular}
\end{table}

For the support-three extension $C_5=T_6\cup\{1100\}$, deleting $0000$, $1001$,
$1010$, or $1100$ leaves an $\cO_6$ tetrahedron; deleting $0111$ leaves the
one-bit-positive four-set whose six pairwise distances are all two.  Thus the
extension table is consistent with the two support-three orbits in
Theorem~\ref{thm:full-k4}.

\FloatBarrier
\section{Three-State Reductions and Minimal-Orbit Values}\label{app:three-state}

A three-state source cannot be treated as three uniform labels: privacy depends
on which subset sums of its positive weights coincide.  We first reduce
the possible weight vectors, then classify the six obstruction orbits.

\subsection{Source-weight reduction}

Let $(p_1,p_2,p_3)$ be the three positive source-state probabilities, with
$p_1+p_2+p_3=1$, and write
$\cU=(1/3,1/3,1/3)$ and $\cQ=(1/2,1/4,1/4)$, up to permutation.
For a local view map $R\to\mathcal A$, a \emph{source block} is a nonempty
fiber: the source states producing one fixed view value.  Its mass is the sum
of their probabilities.
Every proper nonempty source block has mass $p_i$ or $1-p_i$; two different
blocks can have equal mass only if $p_i=p_j$ or $p_i=1/2$.

\begin{proposition}[Three-state collapse outside $\cU,\cQ$]\label{prop:weight-collapse}
For every graph and domain, a scheme with three positive source states whose
weight vector is neither $\cU$ nor $\cQ$ implies a scheme for the same instance
with source support at most two.  Consequently, if $m_G(D)>2$ but
$m_G(D)\le3$, every minimum-support realization has weights $\cU$ or $\cQ$.
\end{proposition}
\begin{proof}
Assume the weight vector is neither $\cU$ nor $\cQ$.
If the weights are distinct and none is $1/2$, different nonempty subsets
of $R$ have different masses.  Thus, at inputs with the same local bit,
privacy forces the blocks producing each view value to contain exactly the
same source states.  Fixing any one source state therefore gives identical
views in every privacy comparison, hence a deterministic private scheme;
correctness survives restriction of the source support.

If exactly two weights agree, then outside $\cU,\cQ$ equal-mass blocks either
both contain the uniquely weighted state or both omit it.  For inputs with
the same local bit, the view value produced by that state is therefore the
same.  Fix that state to obtain a deterministic scheme as above.

In the remaining case, write $R=\{r_\star,\ell_1,\ell_2\}$, where $r_\star$
has weight $1/2$ and $\ell_1,\ell_2$ have distinct positive weights.
The only distinct nonempty source blocks of equal mass are
$\{r_\star\}$ and $\{\ell_1,\ell_2\}$.
Fix a vertex and two inputs with the same local bit.
Outside the partition $\{\{r_\star\},\{\ell_1,\ell_2\}\}$, privacy matches
the view values state by state: blocks with unique masses determine their
source subsets, and these also determine any remaining block.
In particular, for the partition into three singleton blocks, the two
light-state blocks determine the remaining heavy-state block.
In the exceptional partition, privacy may only preserve or exchange the
two mass-$1/2$ view values.
Restrict the source to $R'=\{r_\star,\ell_1\}$ and give both states weight
$1/2$.  In every privacy comparison the two views are unchanged or exchanged,
so their uniform distribution is unchanged, even if the values coincide.
Privacy is therefore preserved, and correctness survives restriction of the
source support.
\end{proof}

\subsection{An exact local-view normal form}

At each input, assign a complete-view label to each source state at each
vertex.  Privacy requires that, for inputs with the same local bit, every
label has the same total probability.  Equal labels at either endpoint force
equal shares on their common edge; taking the equivalence closure of these
requirements records all forced edge identifications.
The final test checks whether the resulting incident edge labels still
distinguish the chosen complete-view labels.

Let $R=\{1,2,3\}$ have a fixed positive weight vector $p$.  Set
$\mathcal X=D\times R$ and $S_v=\{x_v:x\in D\}$.  For each $v$ and
$b\in S_v$, consider disjoint nonempty finite label sets $\mathcal A_{v,b}$ and
surjections $\kappa_{v,x}:R\twoheadrightarrow \mathcal A_{v,x_v}$.  Define
$M_{v,x}(a)=\sum_{r:\kappa_{v,x}(r)=a}p_r$.  For $e=uv$, let $\equiv_e$
be the equivalence relation on $\mathcal X$ generated by
\[
 (x,r)\sim_e(y,s)
 \quad\text{if}\quad
 \kappa_{u,x}(r)=\kappa_{u,y}(s)
 \ \text{or}\ 
 \kappa_{v,x}(r)=\kappa_{v,y}(s),
\]
and set $\sigma_v(x,r)=([(x,r)]_{\equiv_e}:e\ni v)$.

\begin{theorem}[Local-view/forced-edge normal form]\label{thm:three-normal}
A three-state scheme with weights $p$ exists if and only if the preceding data
can be chosen so that
\[
 \begin{gathered}
 M_{v,x}(a)=M_{v,y}(a)
 \quad\text{if }x_v=y_v,\ a\in \mathcal A_{v,x_v},\\
 \sigma_v(x,r)=\Sigma_v\bigl(\kappa_{v,x}(r)\bigr)
 \quad\text{for }(x,r)\in\mathcal X,
 \end{gathered}
\]
where each map
$\Sigma_v:\bigsqcup_{b\in S_v}\mathcal A_{v,b}\to
\prod_{e\ni v}(\mathcal X/{\equiv_e})$ is injective.  In that case canonical
edge shares are $W_e(x,r)=[(x,r)]_{\equiv_e}$.
\end{theorem}
\begin{proof}
Given a scheme, take $\mathcal A_{v,b}$ to be the support of the view distribution at
local bit $b$ and set $\kappa_{v,x}(r)=\View_v(x,r)$.  Privacy gives
the equality of local masses, while correctness makes the two local-bit label
sets disjoint.  Equal $\kappa_v$-labels generate $\equiv_e$ on every incident
edge, so $\sigma_v$ factors through the local label.  The original edge share
is constant on every $\equiv_e$-class; hence equality of two resulting
signatures would make the corresponding original views equal, and the factor
$\Sigma_v$ is injective.

Conversely, use the stated canonical shares.  Their incident view at $v$ is
$\Sigma_v(\kappa_{v,x}(r))$; injectivity makes it equivalent to the local
label.  The mass equality therefore gives privacy, and disjoint label sets for
the two bits give correctness.
\end{proof}

Fix an ordering of $D$.  A \emph{normalized local configuration} at $v$
consists of the label sets $\mathcal A_{v,b}$ and the maps
$\kappa_{v,x}$ for every $x\in D$, with the labels fixed canonically at
the first input of each occurring local-bit class.  The maps must be onto
and satisfy the local mass equalities of Theorem~\ref{thm:three-normal}.
For a three-point source there are only five set partitions.  After choosing
the partition at the first input of a local-bit class, the map at another
input assigns its source blocks to the fixed labels with matching masses;
we call such an assignment a \emph{block-mass-preserving transporter}.
Thus Theorem~\ref{thm:three-normal} gives a finite exact enumeration,
without assuming linearity, a binary share alphabet, or a prescribed
edge partition.

In Table~\ref{tab:search}, a tuple such as $(33,85,85,81)$ lists, in vertex
order $1,2,3,4$, the numbers of normalized local configurations considered
independently at the four vertices, before checking their joint edge
compatibility.

\begin{table}[!t]
\caption{Exact support-three search after the general weight collapse.
Counts give the numbers of normalized local configurations at vertices
$1,2,3,4$ before joint edge-compatibility testing.
``UNSAT'' means that the exhaustive search found no feasible normalized
configuration; ``SAT'' means that it found one.  An UNSAT row is a
computer-assisted theorem conditional only on the proved normal form and
correct execution.}
\label{tab:search}
\centering
\footnotesize
\setlength{\tabcolsep}{3pt}
\begin{tabular}{c c c c c}
\toprule
Orbit & $\cQ$ counts & $\cQ$ result & $\cU$ counts & $\cU$ result\\
\midrule
$\cO_1$ & $33,85,85,85$ & UNSAT & $298,320,320,320$ & UNSAT\\
$\cO_2$ & $33,85,85,81$ & UNSAT & $298,320,320,256$ & UNSAT\\
$\cO_3$ & $33,85,81,81$ & UNSAT & $298,320,256,256$ & UNSAT\\
$\cO_4$ & $85,85,85,81$ & UNSAT & $320,320,320,256$ & UNSAT\\
$\cO_5$ & $81,85,85,81$ & UNSAT & $256,320,320,256$ & UNSAT\\
$\cO_6$ & $81,85,81,81$ & UNSAT & $256,320,256,256$ & SAT\\
\bottomrule
\end{tabular}
\end{table}

\subsection{Minimal-orbit lower bounds}

\begin{theorem}[Support values on minimal obstructions]\label{thm:orbit-values}
For every domain $D$ in the six orbits of Table~\ref{tab:orbits},
\begin{equation}
 m_{K_4}(D)=
 \begin{cases}
 4, & D\in\displaystyle\bigcup_{i=1}^{5}\cO_i,\\
 3, & D\in\cO_6.
 \end{cases}
 \label{eq:orbit-values}
\end{equation}
\end{theorem}
\begin{proof}
Each representative is an odd affine tetrahedron, so
Theorem~\ref{thm:odd-tetra} gives $m>2$.
For the representative of $\cO_1$, Appendix~\ref{app:o1} excludes three states
by a direct argument.  For the representatives of $\cO_2,\ldots,\cO_5$,
Proposition~\ref{prop:weight-collapse} leaves only the source distributions $\cU,\cQ$.
The exhaustive searches in Table~\ref{tab:search} find no accepting branch
under either distribution.  Theorem~\ref{thm:three-normal} makes this exclusion exact:
every scheme induces an enumerated configuration, and every accepted
configuration yields a scheme.  The search instances are described below;
enumeration, soundness, completeness, and execution-audit details are in
Appendix~\ref{app:computation}.  These five representatives therefore require
at least four states.  Proposition~\ref{prop:k4-upper} gives the matching
upper bound.

For the representative $T_6$ of $\cO_6$, restrict the three-state construction
of Theorem~\ref{thm:C5-scheme} from $C_5$ to $T_6$.  Together with $m>2$,
this gives $m_{K_4}(T_6)=3$.  Finally, independent bit negations and coordinate
permutations preserve $m_{K_4}$, so the values hold for every domain in the
respective orbits.
\end{proof}

\subsection{Search instances and results}

For $\cO_1$, Appendix~\ref{app:o1} excludes both remaining source-weight cases by local
partition arguments and the lower bound for Family~11 in the 3SS
classification of Anilkumar et al.~\cite{Anilkumar2023};
Proposition~\ref{prop:k4-upper} supplies the matching upper bound.

For $\cO_2,\ldots,\cO_5$, Proposition~\ref{prop:weight-collapse} leaves only
$\cU,\cQ$, and Theorem~\ref{thm:three-normal} is the exact feasibility test.
For each vertex and local-bit class, normalize the first local map and enumerate
the block-mass-preserving transporters at the other inputs.  Each pair of
endpoint configurations generates the forced closure $\equiv_e$; a tuple is
feasible exactly when its three incident closure signatures separate every pair
of distinct local labels at each vertex.  Table~\ref{tab:search} gives the
vertexwise numbers of normalized configurations and the result in each source-weight case;
Appendix~\ref{app:computation} derives the counts and specifies the exhaustive
implementation.

\FloatBarrier
\section{A Direct Support-Four Lower Bound for \texorpdfstring{$\cO_1$}{O1}}\label{app:o1}

Let
\begin{equation}
 D_1=\{0000,0001,0010,0100\}.
 \label{eq:D1}
\end{equation}
Write the constant coordinate as $h=1$ and the active coordinates as
$a=2,b=3,c=4$.  We call the three edges $ha,hb,hc$ the \emph{helper spokes}.
Deleting $h$ projects to the domain labelled Family~11 in the 3SS
classification of Anilkumar et al.~\cite{Anilkumar2023}:
\begin{equation}
 D_{\mathrm{F11}}=\{000,001,010,100\}.
 \label{eq:F11}
\end{equation}
Its minimum source support is six~\cite{Anilkumar2023}.  This does not
immediately give a lower bound on $K_4$: a constant secret at $h$ does
not make its incident shares constant.  Those shares can supply
correlated randomness needed for reconstruction at the active vertices,
so simply deleting them can destroy correctness.  For this orbit, local
partition constraints let us rule out that possibility under a
three-state source and obtain a direct proof using the Family~11 bound.
Theorem~\ref{thm:k4-onebit} already gives $m_{K_4}(D_1)>2$.  We prove that support three
is impossible without using an exhaustive search.

By Proposition~\ref{prop:weight-collapse}, it remains to exclude only the
uniform source-weight case $\cU$ and the source-weight case $\cQ=(1/2,1/4,1/4)$.
A single-coordinate flip must preserve the other vertices' view distributions
while making the changed vertex's view supports disjoint.  We translate these
requirements into edge constraints and sum them to obtain a contradiction.

\begin{lemma}[Deleting a constant helper]\label{lem:o1-helper-delete}
In a putative three-state scheme for $D_1$, the three helper spokes cannot all
be constant at the base input.
\end{lemma}
\begin{proof}
Otherwise the helper view at the base input is a point mass.  Since $x_h$ is
constant on $D_1$, perfect privacy gives the same point mass at every input;
positivity of all source weights then makes every helper share pointwise
constant.  After deleting $h$, privacy survives marginalization.  A collision
of projected views across two local bits would remain a collision after
reinserting the fixed helper-spoke coordinate, contradicting the original
correctness condition.  We therefore obtain a scheme with at most three source
states for \eqref{eq:F11}, contrary to the Family~11 lower
bound~\cite{Anilkumar2023}.
\end{proof}

\subsection{Uniform source}

Let $R=\{1,2,3\}$ be uniform and fix the base input $0=0000$.
For each edge $e$, define an equivalence relation on the source states by
\[
 r\sim_{e,0}s\quad\Longleftrightarrow\quad W_e(0,r)=W_e(0,s),
 \qquad P_e:=R/{\sim_{e,0}}.
\]
Thus a block of $P_e$ consists of the states giving the same base-input edge
share.  Write $|P_e|$ for its number of blocks and set
\begin{equation}
 d_e:=|P_e|-1\in\{0,1,2\}.
 \label{eq:edge-block-count}
\end{equation}
The three cases are: $d_e=0$ for the one-block or \emph{constant partition}
$P_e=\{R\}$, equivalently $W_e(0,\cdot)$ is constant; $d_e=1$ for a
\emph{two-block partition}; and $d_e=2$ for the \emph{fine partition}
$P_e=\{\{1\},\{2\},\{3\}\}$.

For each active $i\in\{a,b,c\}$, the domain contains the input $e_i$ obtained
from the base input by flipping only coordinate $i$.

\begin{lemma}[Uniform local flip]\label{lem:o1-uniform-flip}
For every active flip $i$ and unchanged vertex $u\ne i$, privacy gives a
permutation $\pi_u\in S_3$ such that
\begin{equation}
 \View_u(e_i,\pi_u(r))=\View_u(0,r)\qquad(r\in R).
 \label{eq:o1-view-perm}
\end{equation}
Consequently, on every edge $uv$ opposite $i$,
\[
 W_{uv}(0,\pi_v^{-1}\pi_u(r))=W_{uv}(0,r)\qquad(r\in R).
\]
Thus $\pi_v^{-1}\pi_u$ preserves each block of $P_{uv}$ individually, and
\begin{equation}
 \sum_{e\ni i}d_e\ge3,\qquad
 \sum_{uv:\,u,v\ne i}d_{uv}\le2.
 \label{eq:o1-local-block-counts}
\end{equation}
More precisely, let the three unchanged vertices be $u_1,u_2,u_3$.
If
\[
 \pi_{u_1}=\pi_{u_2},\qquad
 \pi_{u_3}=\pi_{u_1}\circ\gamma
\]
for a three-cycle $\gamma$ of the source states, then $P_{iu_3}$ is the
fine partition, equivalently $d_{iu_3}=2$.
\end{lemma}
\begin{proof}
Because the source is uniform, equality of the view distributions at $0$ and
$e_i$ means equality of the corresponding multisets of three view values.
A permutation $\pi_u\in S_3$ therefore matches their occurrences, giving
\eqref{eq:o1-view-perm}.  Comparing the two endpoint expressions for
$W_{uv}(e_i,\pi_u(r))$ gives the displayed blockwise invariance.
To examine correctness at the changed vertex $i$, write its base-input and
changed-input views as
\[
\begin{aligned}
 F_i(s)&=(W_{iu}(0,s):u\ne i),\\
 G_i(t)&=(W_{iu}(0,\pi_u^{-1}(t)):u\ne i).
\end{aligned}
\]
The privacy matching at each unchanged $u$ gives
$W_{iu}(e_i,t)=W_{iu}(0,\pi_u^{-1}(t))$, so $G_i(t)$ is exactly the
view at $i$ after the flip, expressed using base-input shares.
Thus $F_i(R)$ and $G_i(R)$ are the supports before and after the flip.
They must be disjoint because the local secret at $i$ changed.

If the sum of the incident $d_e$ values were at most two, either at most one coordinate would
be nonconstant, when all other coordinates are the same constants and its
marginal support makes $F_i(R)$ and $G_i(R)$ meet, or there would be two
coordinates with two-block partitions.  For distinct two-block partitions,
the base and changed joint supports have sizes at least three and two in a
four-point product and must meet.  For equal partitions, either the two
transported singleton states coincide, giving the same diagonal support, or
the third state lies in both doubleton blocks.  This proves the first bound.

Up to common left multiplication, conjugation, and permutation of the
unchanged vertices, the permutation triple is
\[
 (1,1,1),\ (1,1,\tau),\ (1,1,\gamma),\
 (1,\tau_1,\tau_2),\ (1,\gamma,\gamma^{-1}),
\]
where $\tau_1\ne\tau_2$ are transpositions and $\gamma$ is a three-cycle.
The first two types make $F_i$ and $G_i$ meet.  In the remaining types, the
relative permutations around the opposite triangle are, respectively,
identity plus two three-cycles, two transpositions plus one three-cycle, or
three three-cycles.  Invariance of each block under a three-cycle forces the
partition to be constant; under a transposition it forces the two exchanged
states into the same block, so the corresponding $d_e$ is at most one.  The sums of $d_e$ over the opposite triangles
are therefore at most $2,2,0$.  For the final assertion, take any base state $s$ and set
$t=\pi_{u_1}(s)$.  The coordinates toward $u_1,u_2$ agree in $F_i(s)$
and $G_i(t)$; toward $u_3$, the changed coordinate is
$W_{iu_3}(0,\gamma^{-1}(s))$.  If $s$ and $\gamma^{-1}(s)$ belonged
to the same block of $P_{iu_3}$, the two complete views would coincide.
A three-cycle visits all three states, and every partition with fewer than
three blocks puts some such pair in one block.  Hence $P_{iu_3}$ must be
fine, proving $d_{iu_3}=2$.
\end{proof}

Let $\Delta_h=d_{ha}+d_{hb}+d_{hc}$ and $\Delta_\triangle=d_{ab}+d_{ac}+d_{bc}$.
Summing the two bounds in
\eqref{eq:o1-local-block-counts} over the active flips gives
\begin{equation}
 \Delta_h+2\Delta_\triangle\ge9,\qquad 2\Delta_h+\Delta_\triangle\le6,
 \label{eq:o1-global-block-counts}
\end{equation}
and hence $\Delta_h\le1$.  The case $\Delta_h=0$ contradicts
Lemma~\ref{lem:o1-helper-delete}.  If $\Delta_h=1$, relabel so that
$d_{hc}=1$ and $d_{ha}=d_{hb}=0$.  The opposite bounds for $a,b$ give
$d_{bc},d_{ac}\le1$; the incident bounds at $c$ and $a$ then force
$d_{ac}=d_{bc}=1$ and $d_{ab}=2$.  For the flip at $c$, the opposite edges
$ha,hb$ are constant and $ab$ is fine, so $\pi_a=\pi_b$.  If the helper
permutation were equal to them or differed by a transposition, a common or fixed
source state would make the base and changed star supports meet.  It must differ
by a three-cycle, and Lemma~\ref{lem:o1-uniform-flip} then forces $d_{hc}=2$,
contrary to $d_{hc}=1$.  Thus the uniform source-weight case is impossible.

% Start balancing in the first column of the final page, after all floats.
\balance
\subsection{Weights \texorpdfstring{$(1/2,1/4,1/4)$}{(1/2,1/4,1/4)}}

Let $r_\star,\ell_1,\ell_2$ be the source states, with respective masses
$1/2,1/4,1/4$.  Call $\{\{r_\star\},\{\ell_1,\ell_2\}\}$ the
\emph{heavy-versus-light partition}.

\begin{lemma}[Heavy--light structure of a flip]\label{lem:o1-heavy-light}
For every active flip $i$:
\begin{enumerate}[label=(\roman*),leftmargin=*]
\item at least two edges incident to $i$ are nonconstant at the base input;
\item at least two edges in the triangle opposite $i$ are constant at the base
input.
\end{enumerate}
\end{lemma}
\begin{proof}
Inspecting the five partitions of $\{r_\star,\ell_1,\ell_2\}$ shows that a privacy matching
at an unchanged vertex is either induced by a source permutation fixing $r_\star$
(and possibly swapping the light states), or its full base-view partition is
$\{\{r_\star\},\{\ell_1,\ell_2\}\}$ and the matching exchanges the two mass-$1/2$
blocks; call this exchange a \emph{root swap}.  Indeed, fine matchings fix $r_\star$, a
two-block partition isolating one light state can only be transported to the
one isolating the other, and only the displayed partition has another
equal-mass block exchange.

Let $T$ be the set of unchanged vertices using root swaps.  If $T=\varnothing$,
all permutation-induced matchings fix $r_\star$, so the changed star at source state $r_\star$ equals
the base star at $r_\star$, violating correctness.  If all three unchanged vertices
use root swaps, each base full view has partition
$\{\{r_\star\},\{\ell_1,\ell_2\}\}$; every
incident edge is therefore constant or heavy-versus-light.  Swapping all three
roots exchanges the base heavy and light star tuples, leaving their support
unchanged.  Hence $1\le|T|\le2$.

Consider an opposite edge crossing the nontrivial cut $T\mid T^c$.  At the
root-swap endpoint, a nonconstant edge would exchange its heavy and light
labels; at the permutation-matched endpoint, source state $r_\star$ is fixed and its heavy label
is preserved.  These requirements are incompatible, so every crossing edge is
constant.  A nontrivial cut of a triangle has two crossing edges, proving (ii).

If at most one edge incident to $i$ were nonconstant, every base-constant
coordinate would remain the same point mass after the flip, and the sole
nonconstant coordinate would have the same marginal support before and after.
The two star supports would intersect, proving (i).
\end{proof}

Let $n_h$ be the number of constant helper spokes and $n_\triangle$ the number of constant
edges in the active triangle.  Summing parts (i) and (ii) over the three active flips gives
\begin{equation}
 n_h+2n_\triangle\le3,\qquad 2n_h+n_\triangle\ge6.
 \label{eq:o1-heavy-light-count}
\end{equation}
Thus $n_\triangle=0$ and $n_h=3$, contradicting
Lemma~\ref{lem:o1-helper-delete}.

Both possible three-state source-weight cases are impossible.  Therefore
$m_{K_4}(D_1)\ge4$, and Proposition~\ref{prop:k4-upper} gives equality.  This
proves the $\cO_1$ row of Theorem~\ref{thm:orbit-values} without a computer
lower bound.
\section{Exact Computation for the \texorpdfstring{$K_4$}{K4} Classification}\label{app:computation}

\subsection{Normalized local configurations}

For a weight class $\mathcal C$ and $t\ge1$, let $N_t(\mathcal C)$ denote the
number of normalized local configurations for $t$ inputs sharing one fixed
local bit, after the labels of the first input are normalized.  In the uniform
source-weight case, the one-block (coarse) partition contributes one family,
the three two-block partitions each allow three transporters per additional
input, and the fine partition allows six.  Hence
\begin{equation}
 N_t(\cU)=1+3\cdot3^{t-1}+6^{t-1}.
 \label{eq:NtU}
\end{equation}
For $\cQ$, the five canonical partitions contribute one or two transporters,
giving
\begin{equation}
 N_t(\cQ)=1+4\cdot2^{t-1}.
 \label{eq:NtQ}
\end{equation}
Set $N_0=1$.  If a vertex sees $t_0$ domain points with local bit zero and $t_1$
with local bit one, its normalized configuration count is
$N_{t_0}N_{t_1}$, which yields Table~\ref{tab:search}.

\emph{Exact decision test.}
The program realizes each forced closure $\equiv_e$ from
Appendix~\ref{app:three-state} by union--find and stores its separated
local-label pairs as a bitset.  It accepts exactly when the predicate stated
there holds; bitsets accelerate intersections but do not relax the predicate.

\emph{Soundness.}
Every accepted configuration satisfies the local mass equations by construction
and signature injectivity by direct checking.  Theorem~\ref{thm:three-normal}
constructs a correct and private scheme.

\emph{Completeness.}
Given any three-state scheme in $\cU$ or $\cQ$, normalize its local labels at the
first input of each local-bit class.  Privacy supplies one of the enumerated
mass-preserving transporters at every other input.  The original edge shares are
constant on the forced equivalence classes, and correctness implies signature
injectivity.  Hence the enumerator contains an accepting branch whenever a
scheme exists.

The checker \texttt{privacy/normal\_form\_search.cpp} independently regenerates local
configurations and forced edge closures from Theorem~\ref{thm:three-normal},
reproducing every $\cU/\cQ$ entry in Table~\ref{tab:search}.
It is not formally verified; correct compiler and hardware execution remain
assumptions.

\subsection{Cube-orbit and witness audit}

The Python audit \texttt{k4\_geometry.py} identifies the $896$ odd
affine tetrahedra among all $1820$ four-sets and closes their $\cG$-orbits,
of sizes $64,192,192,192,192,64$.  Its scan of all $65{,}535$ nonempty cube
domains finds $80$ containing an $\cO_6$ four-set and no $\cO_1$--$\cO_5$
four-set: the $64$ four-point domains in the orbit of $T_6$ and the $16$
five-point domains in the orbit of $C_5$.
This scan is a cross-check; the proof of Theorem~\ref{thm:full-k4} uses
Appendix~\ref{app:k4-signed}'s extension table.  The audit also verifies that
twelve-row table, the $C_5$ formula, and Proposition~\ref{prop:k4-upper}'s
four-state scheme.
Both programs, the driver \texttt{run\_exact\_search.py}, the logs, and
exact commands are available at
\url{https://github.com/pjmarszalik/price-of-privacy-code}.


\begin{thebibliography}{9}

\bibitem{BlundoMasucci1999}
C.~Blundo and B.~Masucci,
``Randomness in multi-secret sharing schemes,''
\emph{J. Universal Comput. Sci.}, vol.~5, no.~7, pp.~367--389, 1999,
doi: 10.3217/jucs-005-07-0367.

\bibitem{Masucci2006}
B.~Masucci,
``Sharing multiple secrets: Models, schemes and analysis,''
\emph{Des. Codes Cryptogr.}, vol.~39, pp.~89--111, 2006,
doi: 10.1007/s10623-005-2761-1.

\bibitem{Data2016}
D.~Data, V.~M. Prabhakaran, and M.~M. Prabhakaran,
``Communication and randomness lower bounds for secure computation,''
\emph{IEEE Trans. Inf. Theory}, vol.~62, no.~7, pp.~3901--3929, Jul.~2016,
doi: 10.1109/TIT.2016.2568207.

\bibitem{Anilkumar2023}
H.~K.~P. Anilkumar, A.~Rajesh, V.~Narayanan, M.~M. Prabhakaran, and
V.~M. Prabhakaran,
``Randomness requirements for three-secret sharing,'' in
\emph{Proc. 2023 IEEE Int. Symp. Inf. Theory (ISIT)}, 2023, pp.~252--257,
doi: 10.1109/ISIT54713.2023.10206455.

\bibitem{Anilkumar2024}
H.~K.~P. Anilkumar, V.~Narayanan, M.~Prabhakaran, and V.~M. Prabhakaran,
``Randomness in private sequential stateless protocols,'' in
\emph{Proc. ASIACRYPT 2024}, LNCS 15490, pp.~36--66, 2024,
doi: 10.1007/978-981-96-0941-3\_2.

\end{thebibliography}
\end{document}